\documentclass[runningheads]{llncs}
\pdfoutput=1
\usepackage[title]{appendix}
\usepackage{graphicx}
\graphicspath{{./Figures/}}
\usepackage{xcolor}
\usepackage[hidelinks]{hyperref}
\usepackage[all]{hypcap} 
\usepackage{algorithm}
\usepackage[noend]{algpseudocode}
\usepackage{multirow}
\usepackage{amsmath,amssymb}
\let\doendproof\endproof
\renewcommand\endproof{~\hfill\qed\doendproof}
\usepackage{geometry}
\begin{document}
\title{A Characterization of Complexity in Public Goods Games}

\author{Matan Gilboa\orcidID{0009-0009-9643-4105}}

\authorrunning{M. Gilboa}

\institute{University of Oxford
\email{matan.gilboa@cs.ox.ac.uk}}

\maketitle 
\begin{abstract}
We complete the characterization of the computational complexity of equilibrium in public goods games on graphs. In this model, each vertex represents an agent deciding whether to produce a public good, with utility defined by a "best-response pattern" determining the best response to any number of productive neighbors. We prove that the equilibrium problem is NP-complete for every finite non-monotone best-response pattern. This answers the open problem of [Gilboa and Nisan, 2022], and completes the answer to a question raised by [Papadimitriou and Peng, 2021], for all finite best-response patterns.

\end{abstract}

\section{Introduction}
Public goods games describe scenarios where multiple agents face a decision of whether or not to produce some "good", such that producing this good benefits not only themselves, but also other (though not necessarily all) agents. Typically, we consider the good to be costly to produce, and therefore an agent might choose not to produce it, depending on the actions of the agents that affect her. This type of social scenarios can be found in various real-life examples, such as vaccination efforts (an individual pays some personal cost for being vaccinated but she and other people in her proximity gain from it) and research efforts (a research requires many resources, but the researcher benefits from the results along with other researchers in similar areas).
As is common in the literature, to model this we use an undirected graph, where each node represents an agent and an edge between two nodes captures the fact that these nodes directly affect one another by their strategy. As in \cite{Modifications,Parameterized,Directed_Paper,point_out_flaw,general_homogeneous_corrected}, in our model the utility of an agent is completely determined by the number of productive neighbors she has, as well as by her own action. We focus on a specific version of the game which has the following characteristics. Firstly, our strategy space is binary, i.e. an agent can only choose whether or not to produce the good, rather than choose a quantity (we call an agent who produces the good a \textit{productive} agent); secondly, our game is \textit{fully-homogeneous}, meaning that all agents share the same utility function and cost of producing the good; and thirdly, our game is \textit{strict}, which means that an agent has a single best response to any number of productive neighbors she might have (i.e. we do not allow indifference between the actions).

The game is formally defined by some fixed cost $c$ of producing the good, and by some "social" function $X(s_i,n_i)$, which takes into account the boolean strategy of agent $i$ and the number of productive neighbors she has (marked as $s_i$ and $n_i$ respectively), and outputs a number representing how much the agent gains. The utility $u_i$ of agent $i$ is then given by the social function $X(s_i,n_i)$, reduced by the cost $c$ if the agent produces the good, i.e.
$u_i(s_i,n_i)=X(s_i,n_i)-c\cdot s_i$.
However, since any number of productive neighbors yields a unique best response (i.e. the game is \textit{strict}), we can capture the essence of the utility function and the cost using what we call (as in \cite{Gilboa_Nisan_public_goods}), a Best-Response Pattern $T:{\rm I\!N} \rightarrow \{0,1\}$. We think of the Best-Response Pattern as a boolean vector in which the $k^{th}$ entry represents the best response to exactly $k$ productive neighbors. We are interested in the problem of determining the existence of a non-trivial pure Nash equilibrium in these games, which is defined as follows.

\vspace{0.1in}
\noindent
{\bf Equilibrium decision problem in a public goods game:} For a fixed Best-Response Pattern $T:{\rm I\!N} \rightarrow \{0,1\}$, and with an undirected graph $G=(V,E)$ given as input, determine whether there exists a pure non-trivial Nash equilibrium of the public goods game defined by $T$ on $G$, i.e. an assignment $s:V \rightarrow \{0,1\}$ that is not all $0$, such that for every $1\leq i \leq |V|$ we have that 
\[s_i=T[\sum_{j\in N(i)}s_j],\]
where $N(i)$ is the set of neighbors of agent $i$.
\vspace{0.1in}

The first Best-Response Pattern for which this problem was studied was the so-called Best-Shot pattern (where an agent's best response is to produce the good only if she has no productive neighbors, namely $T=[1,0,0,0,...]$), which was shown in \cite{Best_Shot} to have a pure Nash equilibrium in any graph. In \cite{Best_Shot}, they also show algorithmic results for "convex" patterns, which are monotonically increasing (best response is 1 if you have at least $k$ productive neighbors). The question of characterizing the complexity of this problem for all possible patterns was first raised in \cite{Directed_Paper}, where they manage to fully answer an equivalent problem on directed graphs: They show tractability for the All-$1$ pattern, the infinite alternating $[1,0,1,0,...]$ pattern, and all patterns beginning with $0$, and NP-completeness for all other patterns. The open question on undirected graphs was then partially answered in \cite{Gilboa_Nisan_public_goods}, where an efficient algorithm is shown for the pattern $[1,1,0,0,0,...]$, and NP-completeness is established for several classes of non-monotone patterns: Those beginning with $0$ or $11$, or have a prefix of the form $1,0,0,...,0,1,1$. There have been several studies concerning other versions of this problem as well. In \cite{point_out_flaw}, the general version of this problem (where the pattern is part of the input rather than being fixed) was shown to be NP-complete when removing the strictness assumption, (i.e. allowing indifference between actions, such that both 0 and 1 are best responses in certain cases) \footnote{The paper \cite{general_homogeneous_flawed} had an earlier version \cite{general_homogeneous_corrected} which presented a proof for this case as well, but an error in the proof was pointed out in \cite{point_out_flaw}, who then also provided an alternative proof.}. In \cite{general_homogeneous_corrected}, NP-completeness is shown for the general version of the problem in the heterogeneous public goods game, in which the utility function varies between agents. In \cite{Modifications}, they show NP-completeness of the equilibrium problem when restricting the equilibrium to have at least $k$ productive agents, or at least some specific subset of agents. In \cite{Parameterized}, the parameterized complexity of the equilibrium problem is studied, for a number of parameters of the graph on which the game is defined.

In \cite{Gilboa_Nisan_public_goods}, two open problems are suggested regarding the two following patterns: $T_1=[1,1,1,0,0,0,...]$, $T_2=[1,0,1,0,0,0,...]$. $T_1$ has been recently solved in \cite{Monotone_Case}, where they show that all monotonically decreasing patterns can be viewed as potential games, and thus always have a pure Nash equilibrium\footnote{Alternatively, known results about $k$-Dominating and $k$-Independent sets (Theorem 19 in \cite{k_dominating_independent_set}) can be used to prove this.\label{Sigal_Oren_footnote}}. 

There are various instances where non-monotonic patterns are of interest. For example, consider work that necessitates collaboration from $n$ agents or a financial effort that is irrelevant if too few agents contribute but if too many do so it becomes redundant from an agent's perspective. Our main contribution is completing the characterization of the equilibrium decision problem for all finite patterns, by showing that for all non-monotone patterns the problem is NP-complete.

\vspace{0.1in}
\noindent
{\bf Theorem:} For any Best-Response Pattern that is non-monotone and finite (i.e., has a finite number of entries with value 1), the equilibrium decision problem in a public goods game is NP-complete (under Turing reductions).
\vspace{0.1in}

The first step along this way was to prove NP-completeness for the above pattern $T_2$ (which we call the 0-Or-2-Neighbors pattern), namely the second open problem by \cite{Gilboa_Nisan_public_goods}. An alternative proof to this specific problem was obtained independently and concurrently in \cite{Monotone_Case}.

We note that we only focus on finite patterns, which we believe to be more applicable to real-life problems that can be modeled by this game. Nonetheless, we find the characterization of all infinite patterns to be of interest, and this topic remains open, though some initial results can be found in Corollary \ref{cor_after_6_doesnt_matter}. Another interesting open problem is to obtain a similar characterization for the non-strict version of the game, where agents are allowed to be indifferent between the two possible actions.

The rest of this paper is organized as follows. In Section \ref{sec_model} we introduce the formal model and some relevant definitions. We then set out to show hardness of all remaining patterns, dividing them into classes. In Section \ref{sec_0_or_2} we present a solution for the open question from \cite{Gilboa_Nisan_public_goods}, showing hardness of the 0-Or-2-Neighbors Best Response Pattern, and expanding the result to a larger sub-class of patterns that begin with 1,0,1. In Section \ref{sec_semi_shrp_general} we show hardness of all patterns beginning with 1,0,0 (where we also have a slightly more subtle division into sub-classes), and in Section \ref{sec_all_spiked} we show hardness of all patterns beginning with 1,0,1 that were not covered in Section \ref{sec_0_or_2}, thus completing the characterization for all finite patterns. The outline of this paper is also depicted\footnote{Some patterns which start with 1,0 were solved in \cite{Gilboa_Nisan_public_goods}, though for simplicity we omit them from Figure \ref{fig_outline}.} in Figure \ref{fig_outline}.

\begin{figure}
	\includegraphics[width=.9\textwidth]{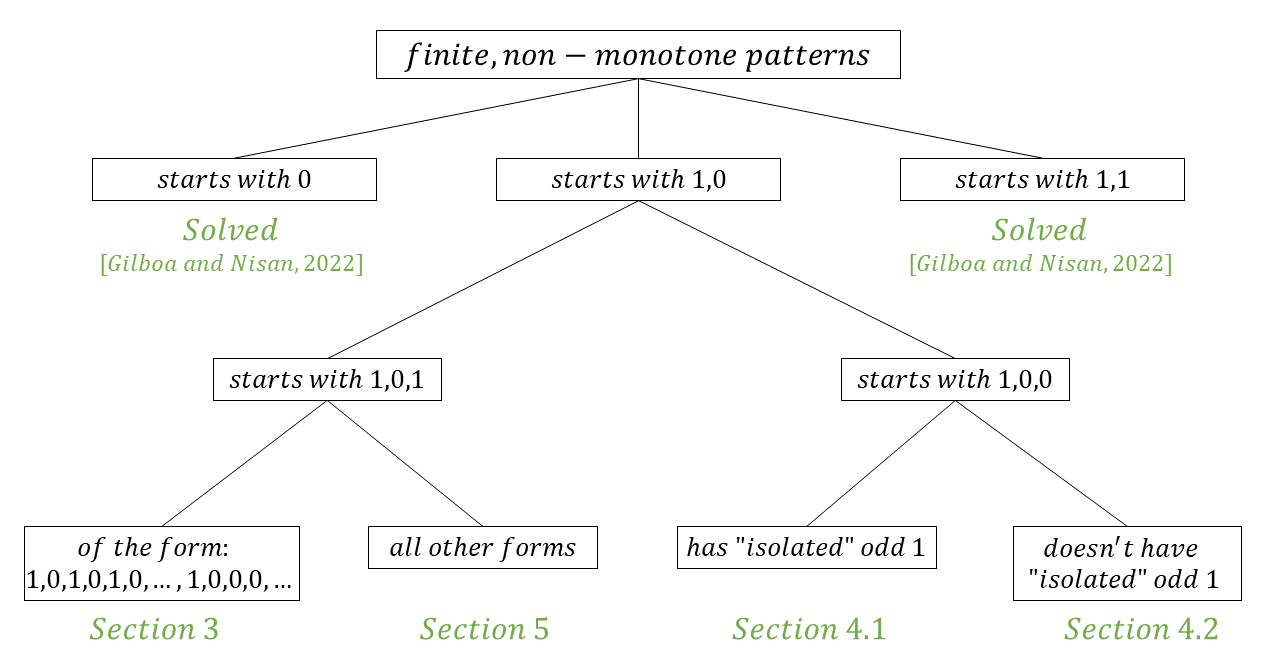}
	\caption{Outline of this paper.}
	\label{fig_outline}
\end{figure}

\section{Model and Definitions}
\label{sec_model}
A \textit{Public Goods Game} (PGG) is defined on an undirected graph $G=(V,E)$, $V=\{v_1,...,v_n\}$, where each node represents an agent. The strategy space, which is identical for all agents, is $S=\{0,1\}$, where 1 represents producing the good and 0 represents not producing it. The utility of node $v_i$ (which is assumed to be the same for all agents) is completely determined by the number of productive neighbors $v_i$ has, as well as by $v_i$'s own strategy. Moreover, our model is restricted to utility functions where an agent always has a single best response to the strategies of its neighbors, i.e. there is no indifference between actions in the game. Therefore, rather than defining a PGG with an explicit utility function and cost for producing the good, we can simply consider the best response of an agent for any number of productive agents in its neighborhood. Essentially, this can be modeled as a function $T:{\rm I\!N}\rightarrow \{0,1\}$, which, as in \cite{Gilboa_Nisan_public_goods}, we represent in the form of a \textit{Best Response Pattern}:
\begin{definition}
A Best-Response Pattern (BRP) of a PGG, denoted by $T$, is an infinite boolean vector in which the $k^{th}$ entry indicates the best response for each agent $v_i$ given that exactly $k$ neighbors of $v_i$ (excluding $v_i$) produce the good:
\begin{align*}
\forall k\geq 0 \;\; T[k]= \text{best response to $k$ productive neighbors.}
\end{align*}
\end{definition}

\begin{definition}
Given a Public Goods Game defined on a graph $G=(V,E)$ with respect to a BRP $T$, a strategy profile $s=(s_1,...,s_n)\in S^n$ (where $s_i\in \{0,1\}$ represents the strategy of node $v_i\in V$) is a pure Nash equilibrium (PNE) if all agents play the best response to the strategies of their neighbors:
\begin{align*}
\forall 1\leq i\leq n\;\; s_i=T[\sum_{j\in N(i)}s_j],
\end{align*}
where $N(i)=\{j:\{v_i,v_j\}\in E\}$.
In addition, if there exists $1\leq i\leq n$ s.t. $s_i=1$, then $s$ is called a non-trivial pure Nash equilibrium (NTPNE).
\end{definition}

We note that throughout the paper we also use the notation $v_i=0$ and $v_i=1$ to indicate the strategy of some node $v_i$, rather than use $s_i=0$ and $s_i=1$, respectively.

\begin{definition}
For a fixed BRP $T$, the non-trivial\footnote{In this paper, we only study BRPs where the best response for zero productive neighbors is 1, for which there never exists a trivial all-zero PNE (as these are the only BRPs left to solve). However, we sometimes reduce from patterns where this is not the case, and therefore include the non-triviality restriction in our problem definition, in order to correspond with the literature.} pure Nash equilibrium decision problem corresponding to $T$, denoted by NTPNE($T$), is defined as follows:
The input is an undirected graph $G$. The output is 'True' if there exists an NTPNE in the PGG defined on $G$ with respect to $T$, and 'False' otherwise.
\end{definition}

\begin{definition}
A BRP $T$ is called monotonically increasing (resp. decreasing) if $\forall k\in {\rm I\!N}$, $T[k]\leq T[k+1]$ (resp. $T[k]\geq T[k+1]$).
\end{definition}

\begin{definition}
A BRP $T$ is called finite if it has a finite number of entries with value 1:
\[\exists N\in {\rm I\!N} \;\;s.t.\;\; \forall n>N \;\; T[n]=0\]
\end{definition}

As seen in Figure \ref{fig_outline}, the only patterns for which the equilibrium decision problem remains open are patterns that begin with 1,0. We divide those into the two following classes of patterns. 

\begin{definition}
A BRP $T$ is called \textit{semi-sharp} if:
\begin{enumerate}
    \item $T[0]=1$
    \item $T[1]=T[2]=0$
\end{enumerate}
i.e. $T$ begins with $1,0,0$.
\end{definition}

\begin{definition}
A BRP $T$ is called \textit{spiked} if:
\begin{enumerate}
    \item $T[0]=T[2]=1$
    \item $T[1]=0$
\end{enumerate}
i.e. $T$ begins with $1,0,1$.
\end{definition}

We note that given any finite BRP $T$, the NTPNE($T$) problem is in NP, since an assignment to the nodes can be easily verified as an NTPNE by iterating over the nodes and checking whether they all play their best response. Therefore, we only prove NP-hardness of the problems throughout the paper.
\section{Hardness of the 0-Or-2-Neighbors Pattern}
\label{sec_0_or_2}
In this section we show that the equilibrium problem is NP-complete for the 0-Or-2-Neighbors pattern, and provide some intuition about the problem. This result answers an open question from \cite{Gilboa_Nisan_public_goods}. We then expand this to show hardness of a slightly more general class of patterns.
In the 0-Or-2-Neighbors BRP the best response is 1 only to zero or two productive neighbors, as we now define.
\begin{definition}
The 0-Or-2-Neighbors Best Response Pattern is defined as follows:
\begin{align*}
\forall k\in {\rm I\!N} \;\; T[k]=    
    \begin{cases}
      1 & \text{if $k=0\;\; or \;\; k=2$}\\
      0  & \text{otherwise}
    \end{cases}
\end{align*}
i.e.
\[T=[1,0,1,0,0,0,...].\]
\end{definition}

\begin{theorem}
\label{thrm_0_or_2}
Let $T$ be the 0-Or-2 Neighbors BRP. Then NTPNE($T$) is NP-complete.
\end{theorem}

Before proving the theorem, we wish to provide basic intuition about the 0-Or-2-Neighbors BRP, by examining several simple graphs. Take for example a simple cycle. Since $T[2]=1$ (i.e. best response for two productive neighbors is 1), we have that any simple cycle admits a pure Nash equilibrium\footnote{In this pattern, any pure Nash equilibrium must also be non-trivial, since $T[0]=1$.}, assigning 1 to all nodes (see Figure \ref{fig_cycle}). However, looking at a simple path with $n$ nodes, we see that the all-ones assignment is never a pure Nash equilibrium. The reason for this is that $T[1]=0$ (i.e. best response for one productive neighbors is 0), and so the two nodes at both ends of the path, having only one productive neighbor, do not play best response. Nevertheless, any simple path does admit a pure Nash equilibrium. To see why, let us observe the three smallest paths, of length 2, 3 and 4. Notice that in a path of length two a PNE is given by the assignment 0,1; in a path of length three a PNE is given by the assignment 0,1,0; and in a path of length four a PNE is given by the assignment 1,0,0,1. We can use these assignment to achieve a PNE in any simple path: given a simple path of length $n$, if $n\equiv0\pmod{3}$ we use the path of length three as our basis, adding 0,1,0 to it as many times as needed; if $n\equiv1\pmod{3}$ we use the path of length four as our basis, adding 0,0,1 to it as many times as needed; and if $n\equiv2\pmod{3}$ we use the path of length two as our basis, adding 0,0,1 to it as many times as needed (see example in Figure \ref{fig_paths}).

\begin{figure}[h!]
\centering
\begin{minipage}[t]{.4\textwidth}
    \includegraphics[width=\textwidth]{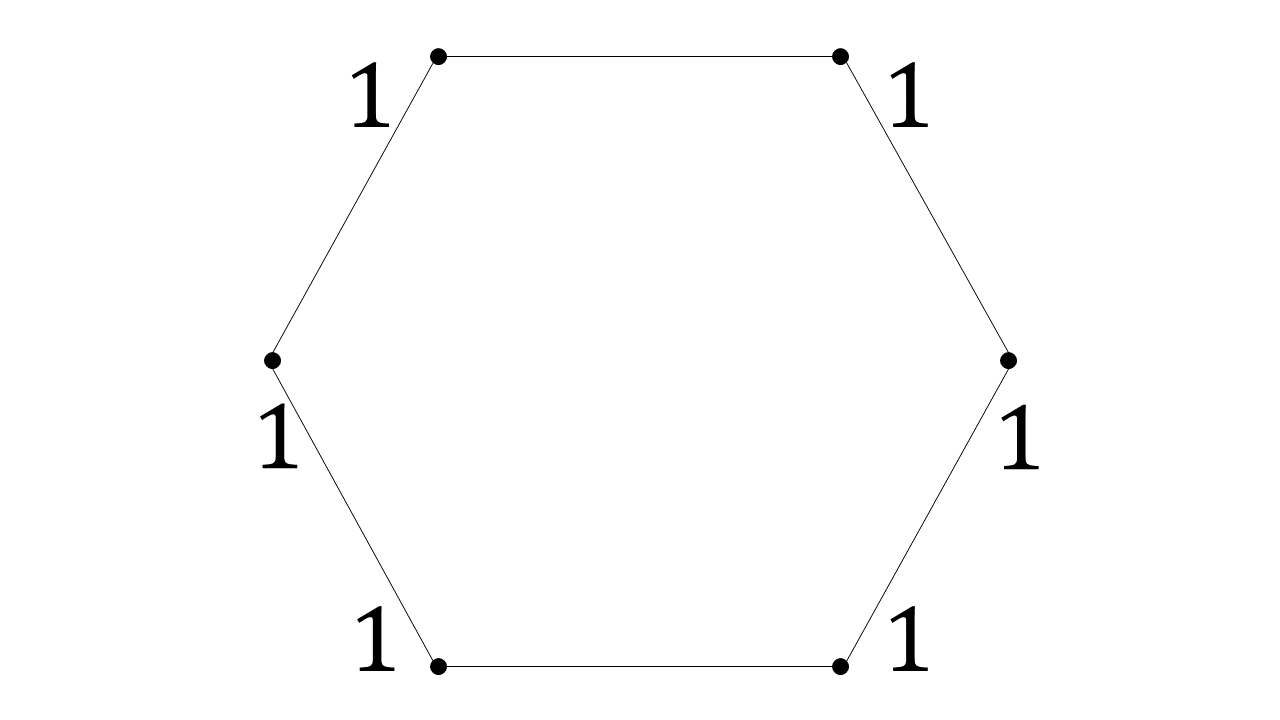}
    \caption{\textit{PNE in cycles.}}
    \label{fig_cycle}
\end{minipage}%
\begin{minipage}[t]{.4\textwidth}
    \includegraphics[width=\textwidth]{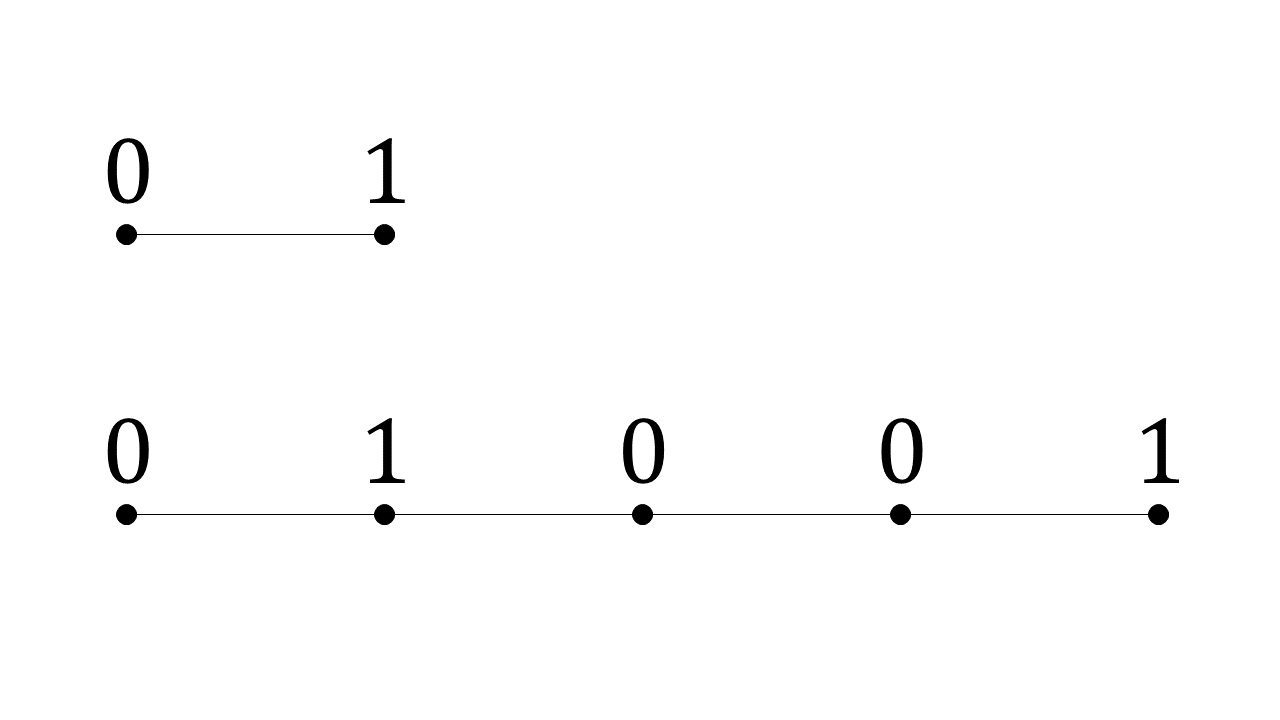}
    \caption{\textit{PNE in paths of lengths 2 and 5.}}
    \label{fig_paths}
\end{minipage}
\end{figure}

In contrast to the graphs discussed so far, there are graphs in which a pure Nash equilibrium doesn't exist for the 0-Or-2-Neighbors pattern. An example of this can be seen in a graph composed of four triangles, connected as a chain where each two neighboring triangles have a single overlapping vertex, as demonstrated in Figure \ref{fig_no_pne}. One may verify that no PNE exists in this graph. This specific structure will also be of use to us during our proof\footnote{The Negation Gadget defined throughout the proof of Theorem \ref{thrm_0_or_2} is constructed similarly to the graph described here.}.

\begin{figure}[h!]
        \centering
	\includegraphics[width=.5\textwidth]{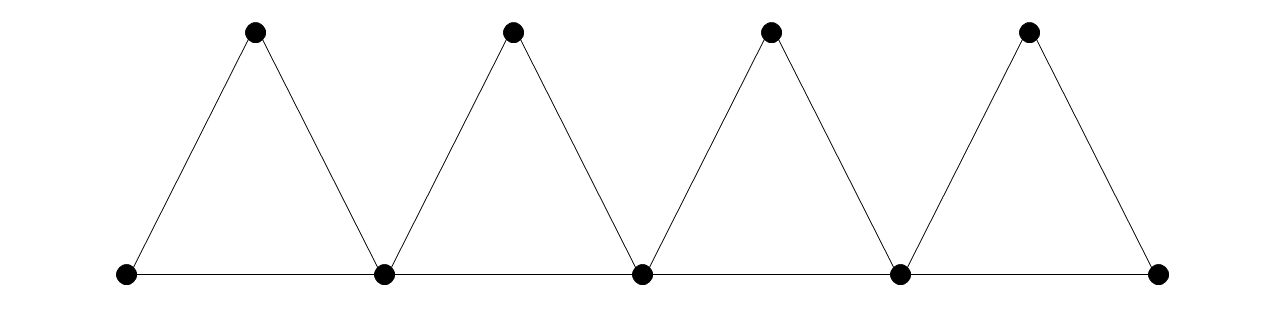}
	\caption{No PNE exists in this graph.}
	\label{fig_no_pne}
\end{figure}

\vspace{0.1in}

Having provided some intuition regarding the problem, we move on to prove Theorem \ref{thrm_0_or_2}. The reduction is from ONE-IN-THREE 3SAT, which is a well known NP-complete problem (see \cite{3sat}). In ONE-IN-THREE 3SAT, the input is a CNF formula with 3 literals in each clause, and the goal is to determine whether there exists a boolean assignment to the variables such that in each clause exactly one of the literals is assigned True.
We begin by introducing our Clause Gadget, which is a main component of the proof. Given a CNF formula, for each of its clauses we construct a 21-nodes Clause Gadget, in which three of the nodes, denoted $l_1, l_2, l_3$ (also referred to as the \textit{literal nodes}) represent the three literals of the matching clause. The purpose of this gadget is to enforce the property that in any NTPNE, exactly one literal node in the gadget will be assigned 1, which easily translates to the key property of a satisfying assignment in the ONE-IN-THREE 3SAT problem. 
The three literal nodes are connected to one another, forming a triangle. Additionally, for each $i\in \{1,2,3\}$, $l_i$ is connected to two other nodes $x_i,y_i$, which are also connected to one another, forming another triangle. Lastly, $x_i$ and $y_i$ each form yet another triangle, along with nodes $a_i,b_i$ and $c_i,d_i$ respectively. We refer to $x_i,y_i,a_i,b_i,c_i,d_i$ as the \textit{sub-gadget} of $l_i$. We note that out of the nodes of the Clause Gadget, only the literal nodes may be connected to other nodes outside of their gadget, a property on which we rely throughout the proof. The structure of the Clause Gadget is demonstrated in Figure \ref{fig_CG}, where each sub-gadget is colored differently.

\begin{figure}
        \centering
	\includegraphics[width=.4\textwidth]{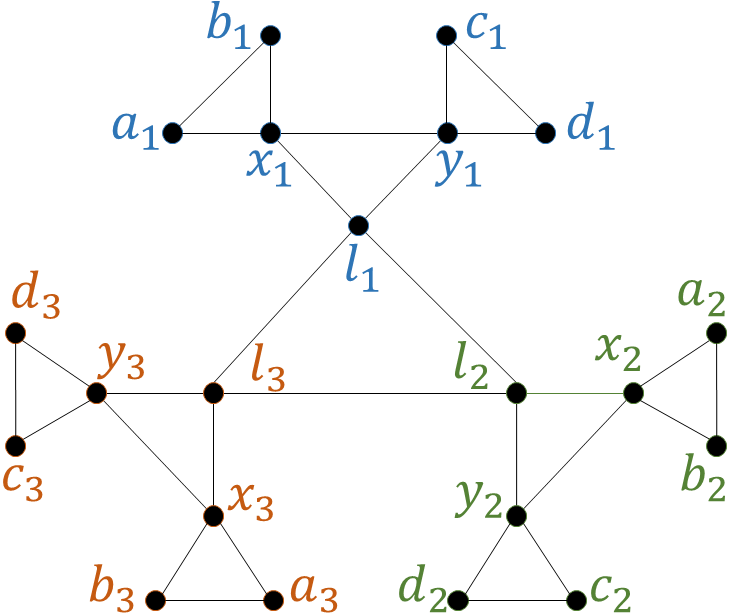}
	\caption{Clause Gadget.}
	\label{fig_CG}
\end{figure}

The next four lemmas lead us to the conclusion that the gadget indeed has the desired property mentioned above.

\begin{lemma}
\label{lem_literal_prod_nbrs}
In any NTPNE in a graph $G$ which includes a Clause Gadget $cg$, if a literal node $l_i$ of $cg$ is assigned 1 then so are its two neighbors from its respective sub-gadget, $x_i,y_i$. Furthermore, there exists an assignment to the sub-gadget of $l_i$ such that all its nodes play best response.
\end{lemma}

\begin{proof}
Divide into cases.

Case 1: If $x_i=y_i=0$, then if $a_i\ne b_i$ (meaning only one of them is assigned 1) then $x_i$ would have two productive neighbors and would not be playing its best response. However, if $a_i=b_i$ then $a_i$ and $b_i$ would not be playing their best response, and we reach a contradiction.

Case 2: If $x_i=1,y_i=0$ (the case where $x_i=0,y_i=1$ is, of course, symmetric) then $x_i$ must have exactly one more productive neighbor (either $a_i$ or $b_i$) in order to be playing best response. But then that node would not be playing best response, in contradiction.

Case 3: We are left with the option where $x_i=y_i=1$, where it is easy to verify that all nodes of the sub-gadget of $l_i$ are playing their best response if we set $a_i=b_i=c_i=d_i=0$.
\end{proof}

\begin{lemma}
\label{lem_one_literal_on}
In any NTPNE in a graph $G$ which includes a Clause Gadget $cg$, if one of the literal nodes $l_i$ of $cg$ is assigned 1 then the other two literal nodes of $cg$ must be assigned 0.
\end{lemma}

\begin{proof}
Since $l_i=1$, from Lemma \ref{lem_literal_prod_nbrs} we have that $x_i=y_i=1$. Therefore, $l_i$ has two productive neighbors and cannot have any more, and so we have that the other two literal nodes must play 0.
\end{proof}

\begin{lemma}
\label{lem_cg_pne}
In any graph $G$ which includes a Clause Gadget $cg$, if exactly one of the literal nodes of $cg$ is assigned 1 then there exists an assignment to the other nodes of $cg$ such that they all (excluding the literal nodes) play best response. In addition, in any such assignment, if the literal nodes have no productive neighbors outside cg, then they also play best response.
\end{lemma}

\begin{proof}
W.l.o.g. assume that $l_1=1, l_2=l_3=0$. Let us observe several details that must hold in such an assignment. Focusing first on the sub-gadget of $l_1$, according to Lemma \ref{lem_literal_prod_nbrs} there exists an assignment the nodes of this sub-gadget such that they all play best response. Furthermore, Lemma \ref{lem_literal_prod_nbrs} tells us that $x_1=y_1=1$ (and so $l_1$ has two productive neighbors within $cg$).
We move on to the sub-gadget of $l_2$. If $x_2\ne y_2$ then $l_2$ would have 2 productive neighbors and would not be playing its best response. If $x_2=y_2=1$ then there is no assignment to $a_2,b_2$ s.t. $a_2,b_2,x_2$ all play their best response. Therefore $x_2=y_2=0$. We are left only with the option of setting $a_2\ne b_2$ and $c_2\ne d_2$ (for instance $a_2=c_2=1, b_2=d_2=0$).
The sub-gadget of $l_3$ is symmetric to that of $l_2$.
One may verify that in this assignment all nodes of $cg$ excluding the literal nodes indeed play best response, and that if the literal nodes have no productive neighbors outside $cg$ then they also play best response.
\end{proof}

\begin{lemma}
\label{lem_all_literals_0}
In any graph $G$ which includes a Clause Gadget $cg$, if all three of the literal nodes of $cg$ are assigned 0, and the literal nodes do not have any productive neighbors outside of $cg$, then the assignment is not a PNE.
\end{lemma}

\begin{proof}
Assume by way of contradiction that there exists a PNE where $l_1=l_2=l_3=0$, and all three of them have no productive neighbors outside $cg$. It must be that the other two neighbors of $l_1$, $x_1,y_1$, are assigned with different values (otherwise $l_1$ is not playing its best response). W.l.o.g. assume $x_1=1,y_1=0$. Now, if the remaining neighbors of $y_1$ ($c_1$ and $d_1$) are both assigned with 0 or both assigned with 1, then they themselves would not be playing their best response. On the other hand, if we assign them with different values then $y_1$ would not be playing its best response, and so we have reached a contradiction.
\end{proof}

So far, we have seen that in any PNE which includes a Clause Gadget, it must be that exactly one of the literal nodes of that gadget is assigned with 1, as long as the literal nodes don't have productive neighbors outside of their Clause Gadget. As we introduce the external nodes that will be connected to the literal nodes, we will show that indeed they all must be assigned 0 in any PNE, and thus a literal node cannot have any productive neighbor outside of its Clause Gadget, which will finalize the property we were looking to achieve with the Clause Gadget.

Our next goal is to make sure the translation between solutions from one domain to the other is always valid. Specifically, we wish to ensure that in any PNE in our constructed graph, if any two literal nodes represent the same variable in the CNF formula then they will be assigned the same value, and if they represent a variable and its negation then they will be assigned opposite values. We begin with the latter, introducing our Negation Gadget. The goal of the Negation Gadget is to force opposite assignments to two chosen nodes, in any Nash equilibrium. The Negation Gadget is composed of 9 nodes: five 'bottom' nodes $b_1,b_2,b_3,b_4,b_5$, and four 'top' nodes $t_1,t_2,t_3,t_4$, and for each $i\leq 4$ we create the edges $\{b_i,b_{i+1}\}$, $\{t_i,b_i\}$ and $\{t_i,b_{i+1}\}$. It can intuitively be described as four triangles that are connected as a chain. Say we have two nodes $u,v$ which we want to force to have opposite assignments, we simply connect them both to node $t_2$ of a Negation Gadget, as demonstrated in Figure \ref{fig_NG}.

\begin{lemma}
\label{lem_neg}
In any Nash equilibrium in a graph $G$ which includes two nodes $u,v$ connected through a Negation Gadget $ng$, $u$ and $v$ must have different assignments. Moreover, the node $t_2$ of $ng$, to which $u$ and $v$ are connected, must be assigned 0. In addition, if indeed $u\neq v$ and $t_2=0$, there exists an assignment to the the nodes of $ng$ such that they all play best response.
\end{lemma}

\begin{proof}
We first show that $u$ and $v$ must have different assignments, dividing into two cases. 

Case 1: Assume by way of contradiction that $u=v=0$. We divide into two sub-cases, where in the first one $t_2=0$; in this case, exactly one of the two remaining neighbors of $t_2$ must be assigned 1 in order for $t_2$ itself to be playing best response. If $b_2=0$ then $b_3=1$ and so, looking at $t_1,b_1$ (the remaining neighbors of $b_2$), we see that any assignment to them results either in $b_2$ not playing best response, or in $t_1$ not playing best response, in contradiction. If, however, $b_3=0$, then $b_2=1$, and so, symmetrically, looking at $t_3,b_4$ (the remaining neighbors of $b_3$) we see that any assignment to them results either in $b_3$ not playing best response, or in $t_3$ not playing best response, in contradiction.
In the second sub-case, where $t_2=1$, we have that its two remaining neighbors must be assigned the same value in order for $t_2$ itself to be playing best response. If $b_2=b_3=0$ then again there is no assignment to $b_1,t_1$ s.t. all of $b_1,t_1,b_2$ play best response, and if $b_2=b_3=1$ then one may verify that there is no assignment to $t_3,t_4,b_4,b_5$ s.t. all of $t_3,t_4,b_3,b_4,b_5$ play best response, and so we reach a contradiction.

Case 2: Assume $u=v=1$. Then we again divide into sub-cases according to $t_2$'s assignment. If $t_2=0$, it must have at least one more productive neighbor in order to play best response. The assignments where $b_2=b_3=1$ or $b_2=0,b_3=1$ are easily disqualified, seeing as there is no assignment to $t_1,b_1$ s.t. $t_1,b_1,b_2$ all play best response. If $b_2=1,b_3=0$ then it must hold that $t_3=b_4$ in order for $b_3$ to play best response, but this would mean that $t_3$ is not playing best response, in contradiction.
If $t_2=1$, then its two remaining neighbors $b_2,b_3$ must be set to 0 in order for it to play best response, and then there is no assignment to $b_1,t_1$ s.t. all of $t_1,b_1,b_2$ play best response, in contradiction.

And so it cannot be that $u=v$. We move on to show that $t_2$ must play 0. Assume by way of contradiction that $t_2=1$. Then, seeing as exactly one of $u,v$ is productive, $t_2$ must have exactly one more productive neighbor in order to play best response. If $b_2=1,b_3=0$ we reach a contradiction as there is no assignment to $t_1,b_1$ s.t. $t_1,b_1,b_2$ all play best response. If $b_2=0,b_3=1$ we reach a contradiction as there is no assignment to $t_3,t_4,b_4,b_5$ s.t. all of $t_3,t_4,b_3,b_4,b_5$ play best response.
Lastly, one may verify that in the assignment where $t_1=b_4=1, b_1=b_2=b_3=b_5=t_2=t_3=t_4=0$ all nodes of the gadget play best response.
\end{proof}

Now, for each variable that appears in the CNF formula, we choose one instance of it and one instance of its negation\footnote{We will soon ensure that instances of the same variable would get the same assignment in any PNE, and thus it is sufficient to negate the assignments of only one instance of a variable and its negation.} and connect the literal nodes representing these instances via a Negation Gadget, thus ensuring they are assigned opposite values in any PNE, according to Lemma \ref{lem_neg}. We note that this is not the only place where we use this gadget, as we will see shortly.

We move on to introduce our Copy Gadget, which we will use to force literal nodes which represent the same variable to have the same assignment in any PNE.
The Copy Gadget is composed of two negation gadgets $ng_1,ng_2$, and two additional nodes $x,y$ which have an edge between them. Say we have two nodes $u,v$ which we want to force to have the same assignment in any PNE, then we simply connect $u$ and $x$ to $ng_1$, and we connect $v$ and $x$ to $ng_2$. The gadget is demonstrated in Figure \ref{fig_CPG}.

\begin{figure}[h!]
\centering
\begin{minipage}[t]{.5\textwidth}
    \includegraphics[width=\textwidth]{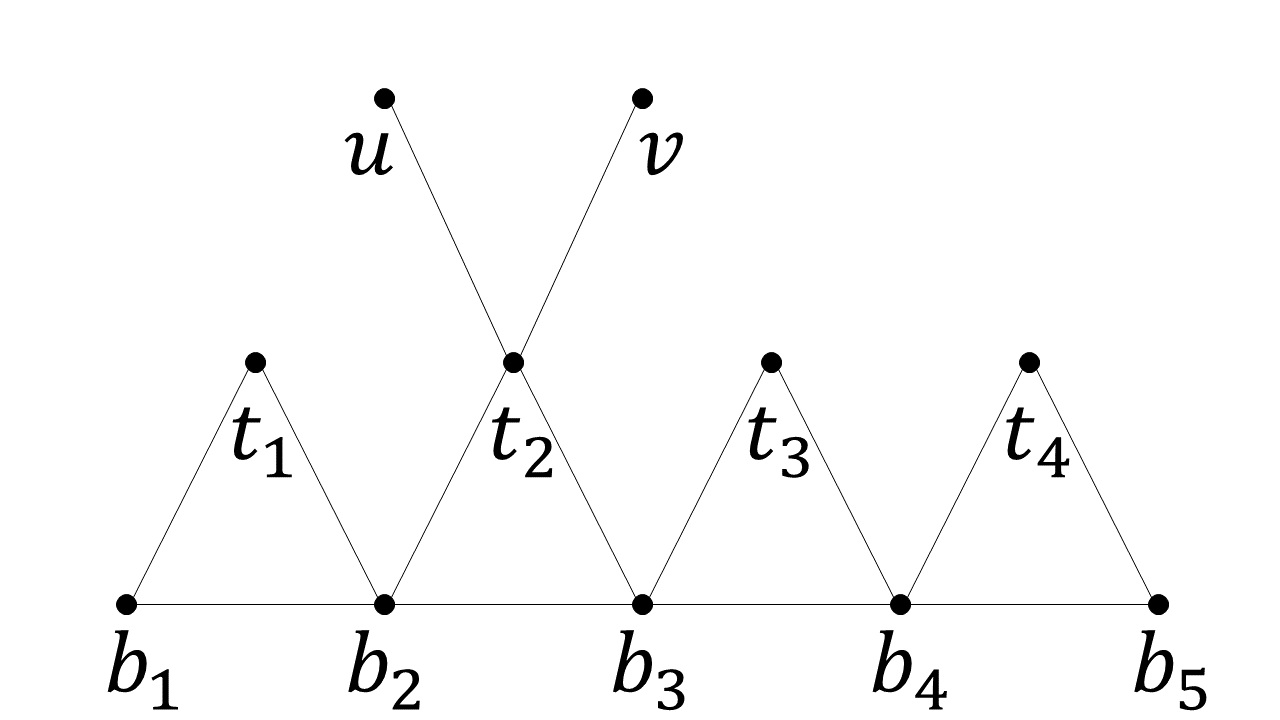}
    \caption{\textit{Negation Gadget connecting $u$ and $v$.}}
    \label{fig_NG}
\end{minipage}%
\begin{minipage}[t]{.5\textwidth}
    \includegraphics[width=\textwidth]{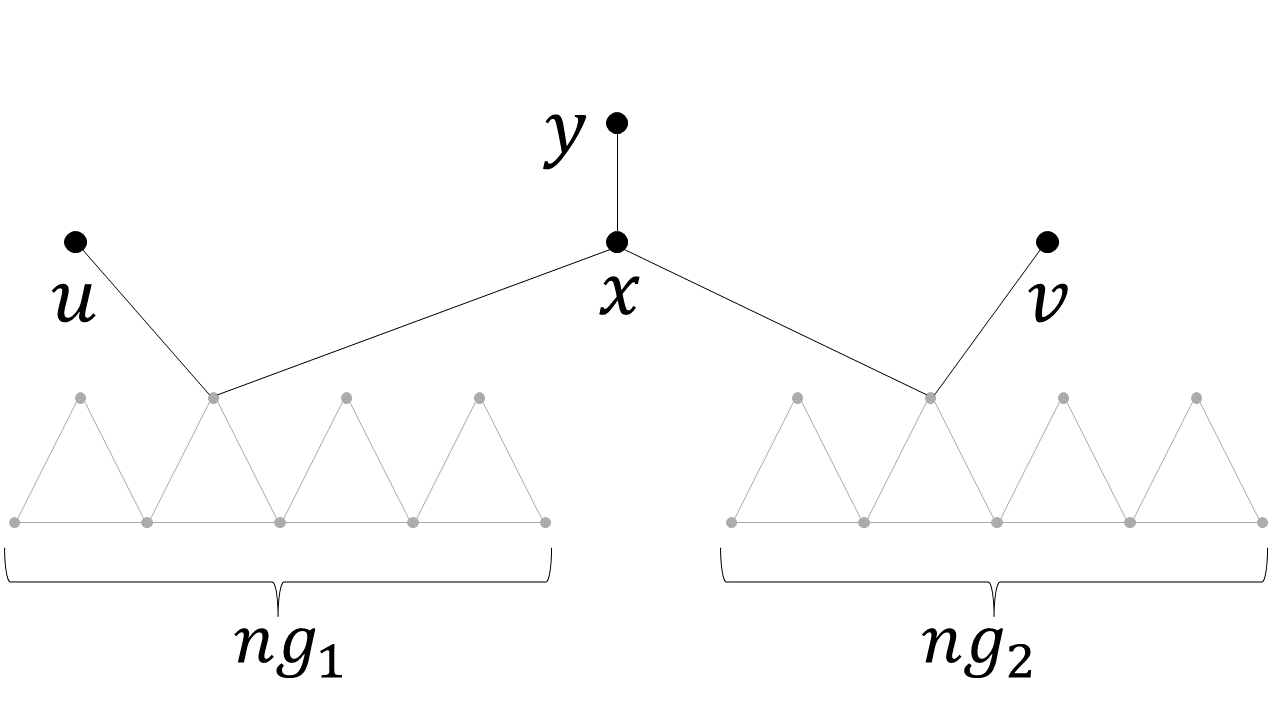}
    \caption{\textit{Copy Gadget connecting $u$ and $v$.}}
    \label{fig_CPG}
\end{minipage}
\end{figure}

\begin{lemma}
\label{lem_cpg}
In any Nash equilibrium in a graph $G$ which includes two nodes $u,v$ connected through a Copy Gadget $cpg$, $u,v$ must have the same assignment, and must have no productive neighbors from $cpg$. In addition, if $u=v$ then there exists an assignment to the nodes of $cpg$ s.t. all of them play best response.
\end{lemma}

\begin{proof}
We first show that $u$ and $v$ must have the same assignment. This follows directly from the fact that $x$ is connected to both $u$ and $v$ via a Negation Gadget. Therefore, from Lemma \ref{lem_neg} we have that $x\ne u$ and $x\ne v$, and so $u=v$. 
Lemma \ref{lem_neg} also tells us that the Negation Gadget cannot add productive neighbors to the nodes that are connected to it in any PNE, and therefore $u$ and $v$ have no productive neighbors from $cpg$.
Lastly, we show that there exists an assignment to the nodes of $cpg$ s.t. they all play best response. From Lemma \ref{lem_neg} $x$ cannot have any productive neighbors from $ng_1$ or $ng_2$. Therefore, if $u=v=0$ then we can assign $x=1,y=0$, and if $u=v=1$ then we can assign $x=0,y=1$. In both cases, we assign $ng_1$ and $ng_2$ as suggested in the proof of Lemma \ref{lem_neg}. One may verify that in this assignment indeed all nodes of $cpg$ play best response.
\end{proof}

Now, for each variable in the CNF formula, we connect all the literal nodes representing its different instances via a chain of copy gadgets, thus (transitively) ensuring they are all assigned the same value in any PNE, according to Lemma \ref{lem_cpg}.

Given these lemmas and the graph we constructed, we can now prove Theorem \ref{thrm_0_or_2}.

\begin{proof}(Theorem \ref{thrm_0_or_2})
Given a ONE-IN-THREE 3SAT instance, we construct a graph $G$ as described previously\footnote{Note that we do not need to explicitly represent the conjunction between the clauses: It is given to us implicitly by the fact that each Clause Gadget must independently satisfy the 1-In-3 property in any Nash equilibrium.}: For each clause we create a Clause Gadget, we connect all literal nodes representing the same variable through a chain of Copy Gadgets, and for each variable we choose one instance of it and one instance of its negation, and connect the literal nodes representing those instances with a Negation Gadget. If there exists a satisfying assignment to the 3SAT problem, we can set all literal nodes according to the assignment of their matching variable, and set all other nodes as described throughout Lemmas \ref{lem_cg_pne}, \ref{lem_neg} and \ref{lem_cpg}, and according to those lemmas, we get a pure Nash equilibrium.
On the opposite direction, if there exists a non-trivial pure Nash equilibrium, then by Lemmas \ref{lem_one_literal_on} and \ref{lem_all_literals_0} in each Clause Gadget exactly one literal node is assigned 1, and by Lemmas \ref{lem_neg} and \ref{lem_cpg} we have that literal nodes have the same assignment if they represent the same variable, and opposite ones if they represent a variable and its negation. Thus we can translate the NTPNE into a satisfying ONE-IN-THREE 3SAT assignment, assigning 'True' to variables whose literal nodes are set to 1, and 'False' otherwise.
\end{proof}

We now wish to expand this result to two slightly more general classes of patterns.
Firstly, we notice that the graph constructed throughout the proof of Theorem \ref{thrm_0_or_2} is bounded\footnote{A literal node is connected to 4 nodes within its clause gadget, and possibly 2 nodes from copy gadgets or 1 node from a negation gadget and 1 node from a copy gadget (assuming we connect the negation gadgets at the end of their respective Copy-Gadget-chains).} by a maximum degree of 6. Therefore, the proof is indifferent to entries of the pattern from index 7 onward, which means it holds for any pattern that agrees with the first 7 entries of the 0-Or-2-Neighbors pattern.

\begin{corollary}
\label{cor_after_6_doesnt_matter}
Let $T$ be a BRP such that:
\begin{itemize}
    \item $T[0]=T[2]=1$
    \item $\forall k\in \{1,3,4,5,6\} \;\; T[k]=0$
\end{itemize}
Then NTPNE($T$) is NP-complete.
\end{corollary}

Secondly, according to Theorem 7 in \cite{Gilboa_Nisan_public_goods}, adding 1,0 at the beginning of a hard pattern that begins with 1 yields yet another hard pattern. Using this theorem recursively on the patterns of Corollary \ref{cor_after_6_doesnt_matter}, we have that the equilibrium decision problem is hard for any pattern of the form:
\[T=[\underbrace{1,0,1,0,1,0,...,1,0}_{finite\;number\;of\;'1,0'},0,0,0,?,?,...].\]

\begin{corollary}
\label{cor_alternating}
Fix $m\geq 1$, and let $T$ be a BRP such that:
\begin{itemize}
    \item $\forall 0\leq k\leq m$ 
    \begin{enumerate}
      \item $T[2k]=1$
      \item $T[2k+1]=0$
    \end{enumerate}
    \item $T[2m+2]=T[2m+3]=T[2m+4]=0$    
\end{itemize}
Then NTPNE($T$) is NP-complete.
\end{corollary}

We will see later on that this result will also be of use during the proof of Theorem \ref{thrm_all_spiked}.

There is one very similar class of patterns on which the proofs throughout the paper rely. This is the class of all finite patterns that start with a finite number of 1,0, followed by 1,1, i.e. all patterns of the form:
\[T=[\underbrace{1,0,1,0,1,0,...,1,0}_{finite\;number\;of\;'1,0'},1,1,?,?,...,0,0,...].\]
The complexity of those patterns was already discussed and solved in Section 5.4 of \cite{Gilboa_Nisan_public_goods}, but was not formalized and so we state it here in the following lemma. 
\begin{lemma}
\label{lem_alternating_with_odd}
Fix $m\geq 2$, and let $T$ be a BRP s.t.:
    \begin{itemize}
        \item $T$ is finite
        \item $\forall k\in {\rm I\!N} \;\;s.t.\;\; 2k\leq m \;\; T[2k]=1$
        \item $T[1]=0$
        \item $\exists 1\leq n \;\;where\;\; 2n+1\leq m+1, \;\;s.t.\;\; T[2n+1]=1$
    \end{itemize}
Then NTPNE($T$) is NP-complete under Turing reduction.
\end{lemma}

\begin{proof}
The proof follows directly from Theorems 6 and 7 from \cite{Gilboa_Nisan_public_goods}.   
\end{proof}

\section{Hardness of Semi-Sharp Patterns}
\label{sec_semi_shrp_general}
In this section we show hardness of semi-sharp Best-Response Patterns, beginning with a specific sub-class of those patterns in Section \ref{sec_semi_shrp_isolated}, and expanding to all other semi-sharp patterns in Section \ref{sec_all_semi_shrp}. We remind the reader that semi-sharp patterns are patterns that begin with 1,0,0.

\subsection{Semi-Sharp Patterns with Isolated Odd 1}
\label{sec_semi_shrp_isolated}
In this section we show that any finite, semi-sharp pattern such that there exists some 'isolated' 1 (namely it has a zero right before and after it) at an odd index, presents a hard equilibrium decision problem. Those patterns can be summarized by the following form:
\[T=[1,0,0,?,?,...,0,\underbrace{1}_{odd\; index},0,?,?,...,0,0,0,...]\]

\begin{theorem}
\label{thrm_semi_sharp_isolated_odd}
Let $T$ be a BRP which satisfies the following conditions:
\begin{itemize}
\item $T$ is finite
\item $T$ is semi-sharp
\item $\exists m\geq 1 \;\;$ s.t.:
    \begin{enumerate}
        \item $T[2m]=T[2m+2]=0$
        \item  $T[2m+1]=1$
    \end{enumerate}
\end{itemize}
Then NTPNE($T$) is NP-complete under Turing reduction.
\end{theorem}

Before we proceed to the proof, we introduce two gadgets and prove two lemmas regarding their functionality.

\textit{Force-1-Gadget}: The first gadget is denoted the Force-1-Gadget, and it will appear in several parts of the graph we construct for the reduction. The goal of this gadget is to enable us to force any node to be assigned 1 in any Nash equilibrium in a PGG defined by $T$.
This gadget is composed primarily of a triangle $x,y,z$, where the triangle's nodes have also several 'Antenna' nodes, which are connected only to their respective node from the triangle. Specifically, $x$ will have $2m+1$ Antenna nodes, and $y$ and $z$ will each have $2m$ Antenna nodes. Say we have some node $u$, whose assignment we wish to force to be 1, then we simply connect $u$ to one of the Antenna nodes of $x$, denoted $a$. The gadget is demonstrated in Figure \ref{fig_FG}.

\textit{Add-1-Gadget}: our second gadget of this proof is denoted the Add-1-Gadget, and its goal is to enable us to assure the existence of (at least) a single productive neighbor to any node in a Nash equilibrium of a PGG defined by $T$.
Say we have a node $v$, to which we wish to add a single productive neighbor, in any equilibrium.
We construct the Add-1-Gadget as follows. We create $m+1$ nodes denoted $x_1,...,x_{m+1}$, $m+1$ nodes denoted $y_1,...y_{m+1}$, and an additional 'bridge' node, denoted $b$. We connect $x_1$ and $y_1$ to all of the other $x_i$ and $y_i$ nodes. For all $i,j\geq 2$ s.t. $i\ne j$, we create the edges $\{x_i,x_j\},\{y_i,y_j\},\{x_i,y_j\}$ (the $x_i,y_i$ nodes almost form a clique, except that for each $i\geq 2$ we omit the edge $\{x_i,y_i\}$). Additionally, for all $i\geq 2$ the bridge node $b$ is connected to $x_i$ and to $y_i$.
To $b$ we attach a Force-1-Gadget, and we also connect $b$ to $v$. The gadget is demonstrated in Figure \ref{fig_AG}.

\begin{figure}[h!]
\centering
\begin{minipage}[t]{.5\textwidth}
    \includegraphics[width=\textwidth]{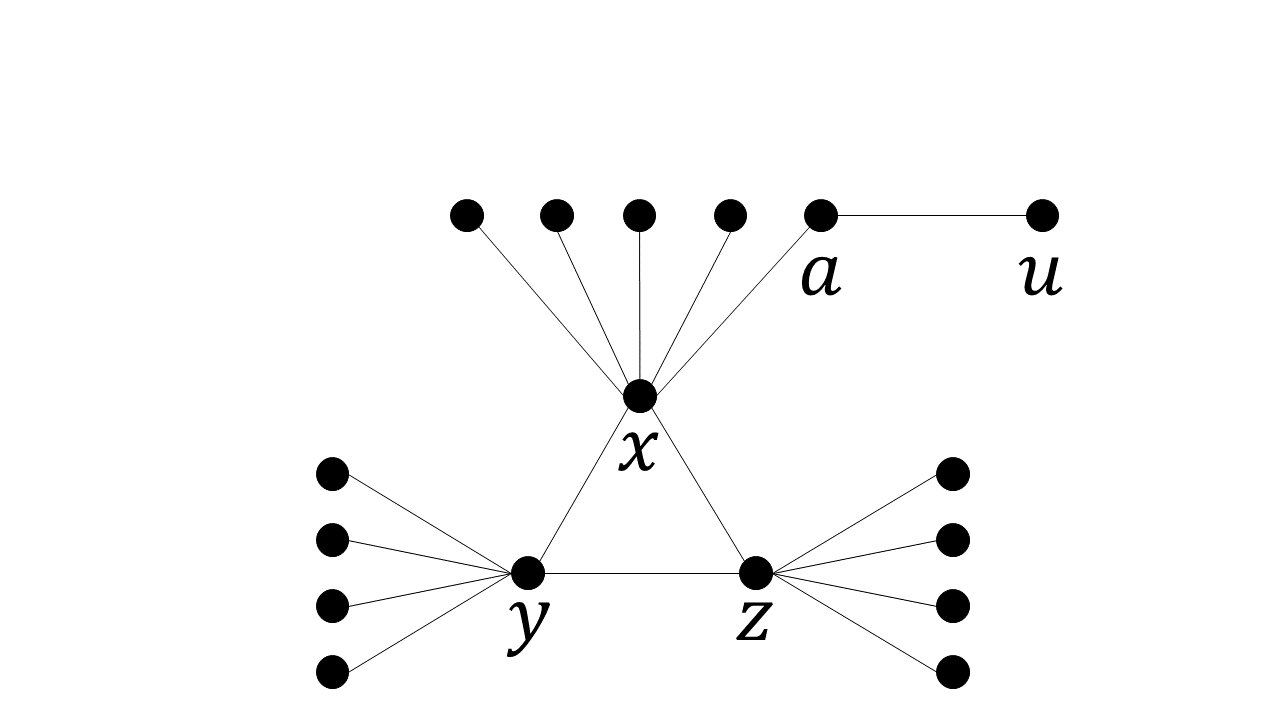}
    \caption{\textit{Force-1-Gadget with $m=2$, attached to $u$.}}
    \label{fig_FG}
\end{minipage}%
\begin{minipage}[t]{.5\textwidth}
    \includegraphics[width=\textwidth]{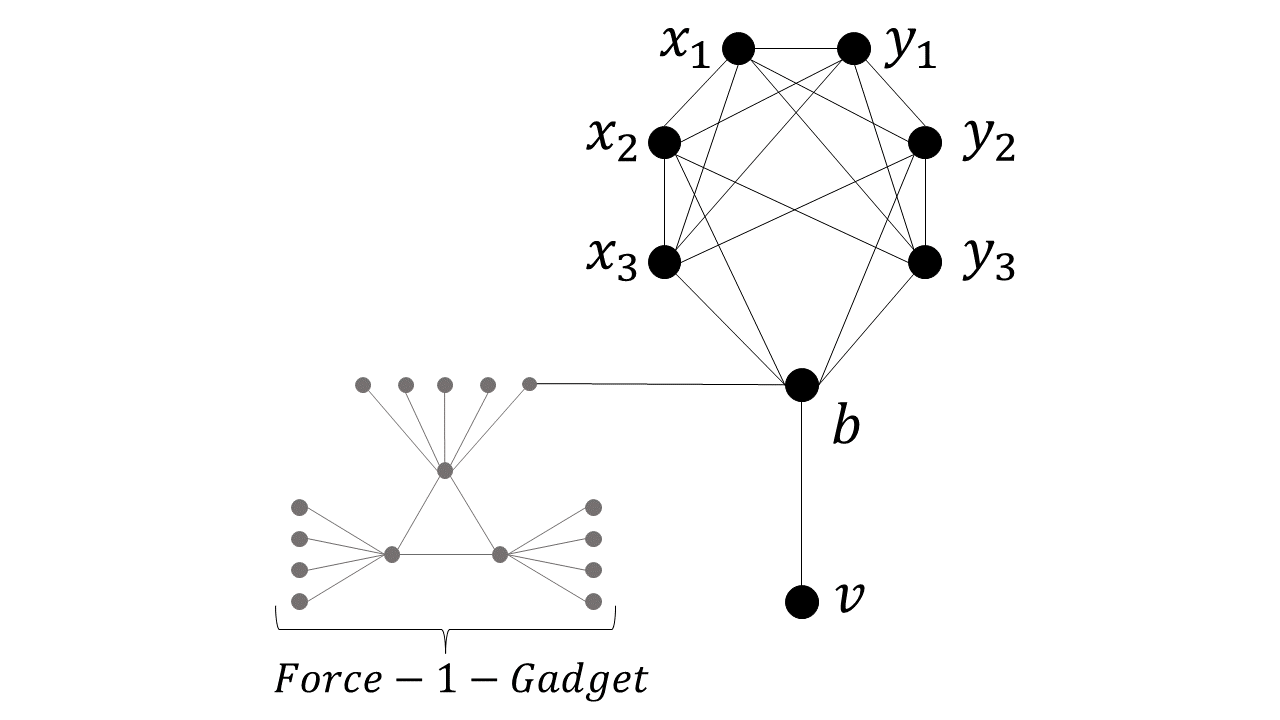}
    \caption{\textit{Add-1-Gadget with $m=2$, attached to $v$.}}
    \label{fig_AG}
\end{minipage}
\end{figure}

The following lemmas formalize the functionality of the two gadgets, beginning with the Force-1-Gadget in Lemma \ref{lem_fg}.

\begin{lemma}
\label{lem_fg}
In any PNE in a graph $G$ corresponding to the BRP $T$ (from Theorem \ref{thrm_semi_sharp_isolated_odd}), where $G$ has a node $u$ that is connected to a Force-1-Gadget $fg$ as described, $u$ must be assigned 1, and its neighbor from $fg$, $a$, must be assigned 0.\footnote{The property that $a=0$ allows us to use the Force-1-Gadget without risking potentially adding productive neighbors to the respective node.} Furthermore, if $u=1$ there exists an assignment to the nodes of $fg$ such that they each play their best response. 
\end{lemma}

\begin{proof}
First we show that $u$ must be assigned 1. Assume by way of contradiction that $u=0$. Divide into the following two cases.
If $x=1$, then all of its Antenna nodes must be assigned 0 (according to $T$). Additionally, $y$ and $z$ must also be assigned 0, as otherwise $x$ wouldn't be playing best response, since $T$ is semi-sharp. Therefore, the best response of all of the Antenna nodes of $y$ and $z$ is to play 1, which leaves $y$ and $z$ with $2m+1$ productive neighbors each, and so they are not playing best response, in contradiction.
If $x=0$, then all of its Antenna nodes must play 1. Therefore, $x$ must have at least one other productive neighbor, as otherwise it would have $2m+1$ productive neighbors and wouldn't be playing best response; w.l.o.g. assume $y=1$. Then all of $y$'s Antenna nodes must play 0. Therefore, $z$ must play 0, as otherwise $y$ wouldn't be playing best response. This means the best response for $z$'s Antenna nodes is to play 1, which leaves $z$ with $2m+1$ productive neighbors, and so it isn't playing best response, in contradiction.
We move on to showing that $a$ must play 0. This follows directly from the fact that $u=1$. Since $a$ only has one other neighbor ($x$), regardless of its strategy the best response for $a$, according to $T$, would be playing 0.
It is left to show that when $u=1$ and $a=0$, there exists an assignment to the nodes of $fg$ s.t. they all play best response. One may verify that when we set $x=y=z=0$ and set all the Antenna nodes in $fg$ (except for $a$) to 1, then all nodes of $fg$ play best response (specifically, $x,y,z$ would each have exactly $2m$ productive neighbors, which, by definition of $T$, means they are playing best response).
\end{proof}

We move on to proving the following Lemma, which formalizes the functionality of the Add-1-Gadget.

\begin{lemma}
\label{lem_ag}
In any graph $G$ corresponding to the BRP $T$ (from Theorem \ref{thrm_semi_sharp_isolated_odd}), where $G$ has a node $v$ that is connected to an Add-1-Gadget $ag$ as described, there always exists an assignment to the nodes of $ag$ such that they all play best response, regardless of $v$'s strategy. In addition, the bridge node $b$ of $ag$ must be assigned 1 in such an assignment.
\end{lemma}

\begin{proof}
The claim that $b$ must play 1 follows directly from the fact that it has a Force-1-Gadget attached to it, i.e. from Lemma \ref{lem_fg}. Additionally, all the nodes of the Force-1-Gadget attached to $b$ can be assigned as suggested in Lemma \ref{lem_fg}. It is left to show a possible assignment to the rest of the nodes of $ag$. We divide into cases.
If $v=0$, then we set $x_1=1$ and all other $x_i,y_i$ nodes we set to 0.
If $v=1$, then we set $x_i=y_i=1$ for all $1\leq i \leq m+1$. One may verify that given these assignments all nodes of $ag$ play their best response.
\end{proof}

Given these two gadgets, we are almost ready to prove Theorem \ref{thrm_semi_sharp_isolated_odd}. We now introduce a useful definition, and then proceed to the proof of the theorem.
\begin{definition}
Let $T$ and $T'$ be two BRPs. We say that $T'$ is \textit{shifted left by $t$} from $T$ if
\[\forall k\geq 0 \;\; T'[k]=T[k+t].\]
\end{definition}

\begin{proof}
(Theorem \ref{thrm_semi_sharp_isolated_odd}) Denote by $T'$ the pattern which is shifted left by 1 from $T$, i.e.:
\[\forall k\geq 0 \;\; T'[k]=T[k+1].\]
Notice that $T'$ is non-monotonic, finite, and begins with $0$, and therefore NTPNE($T'$) is NP-complete according to Theorem 4 in \cite{Gilboa_Nisan_public_goods}, which allows us to construct a Turing reduction from it.
The technique of the reduction is very similar to those of the proofs of Theorems 5 and 6 in \cite{Gilboa_Nisan_public_goods}.
Given any graph $G=(V,E)$, where $V={v_1,...,v_n}$, we construct $n$ graphs $G_1,...,G_n$, where for each $1\leq i\leq n$ the graph $G_i$ is defined as follows. The graph $G_i$ contains the original input graph $G$, and in addition, we connect a unique Add-1-Gadget to each of the original nodes, and a Force-1-Gadget only to node $v_i$. If there exists some non-trivial PNE in the PGG defined on $G$ by $T'$, let $v_j$ be some node who plays 1. Then the same NTPNE is also an NTPNE in the PGG defined by $T$ on $G_j$, when we assign the nodes of the additional gadget as suggested in Lemmas \ref{lem_fg} and \ref{lem_ag}. To see why, notice that $T'$ is shifted left by 1 from $T$, and the Add-1-Gadgets ensure that all nodes have exactly one additional productive neighbor than they had in $G$.

In the other direction, if there exists an NTPNE in a PGG defined by $T$ on one of the graphs $G_j$, then by the same logic this is also a PNE in the game defined by $T'$ on $G$ (ignoring the assignments of the added nodes). Moreover, the Force-1-Gadget ensures this assignment is non-trivial even after removing the added nodes, since $v_j$ must play 1 in this assignment.
\end{proof}

\subsection{All Semi-Sharp Patterns}
\label{sec_all_semi_shrp}
In this section we show that any finite, non-monotone, semi-sharp pattern presents a hard equilibrium problem.

\begin{theorem}
\label{thrm_all_semi_shrp}
Let $T_1$ be a finite, non-monotone, semi-sharp BRP. Then NTPNE($T_1$) is NP-complete under Turing reduction.
\end{theorem}

Before proceeding to the proof, we wish to introduce the following definition and prove two lemmas related to it.

\begin{definition}
Let $T$ and $T'$ be two BRPs such that $\forall k\in {\rm I\!N}$ it holds that $T[k]=T'[2k]$. Then we say that $T'$ is a \textit{double-pattern} of $T$, and $T$ is the \textit{half-pattern} of $T'$. Notice that a pattern has a unique half-pattern, whereas, since the definition does not restrict $T'$ in the odd indices, any pattern has infinite double-patterns.
\end{definition}

The first lemma is very simple and intuitive, stating that the largest index with value 1 in a half pattern is strictly smaller than the largest index with value 1 in its original pattern. This is true since for any index $i$ s.t. the value of the half pattern is 1 in that index, the original pattern has a value of 1 in index $2i$.

\begin{lemma}
\label{lem_half_reaches_best_shot}
Let $T$ and $T'$ be two finite BRPs such that $T$ is the half-pattern of $T'$. Denote by $i$ the largest index s.t. $T[i]=1$ and denote by $j$ the largest index s.t. $T'[j]=1$. Then if $j>0$ we have that $i<j$.
\end{lemma}

\begin{proof}
The proof is trivially given by the definition of a half pattern, since $T'[2i]=T[i]$.
\end{proof}

The next lemma is less trivial, stating the relation between hardness of a pattern and its double-pattern.
\begin{lemma}
\label{lem_double_pat}
Let $T$ be a BRP such that NTPNE($T$) is NP-complete, and let $T'$ be a double-pattern of $T$. Then NTPNE($T'$) is NP-complete.
\end{lemma}

\begin{proof}
We use a specific case of the same reduction that was used to prove Theorem 4 in \cite{Gilboa_Nisan_public_goods}.
Given a graph $G_1=(V_1,E_1)$ as input, where $V_1={v_1^1,...,v_n^1}$, we create another replica of it $G_2=(V_2,E_2)$, where $V_2={v_1^2,...,v_n^2}$. For each node (from both graphs), we add edges connecting it to all replicas of its neighbors from the opposite graph. That is, the following group of edges is added to the graph:
\[E=\{\{v_i^1,v_j^2\}|\{v_i^1,v_j^1\}\in E_1\}.\]
A demonstration of the reduction can be seen in Figure \ref{fig_lem_double_pat}.
\begin{figure}[h]
        \centering
	\includegraphics[width=.5\textwidth]{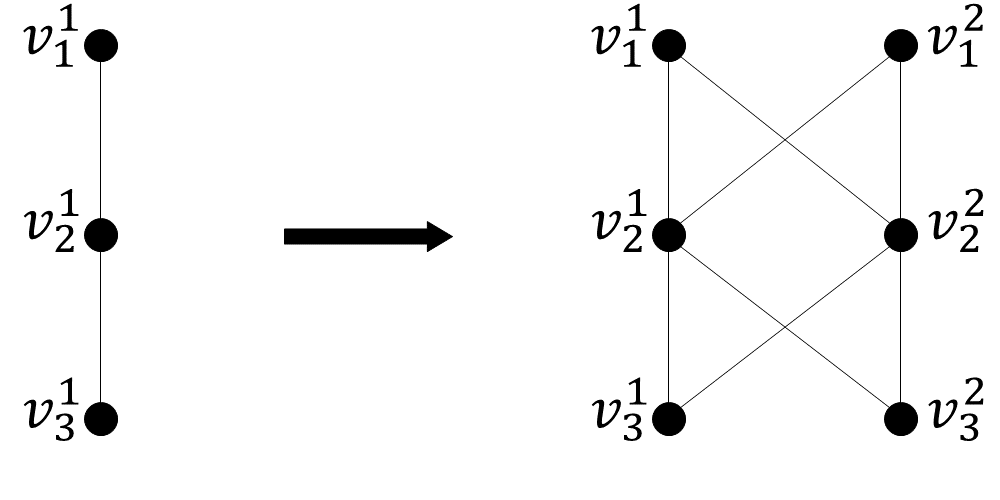}
	\caption{Example of the reduction of Lemma \ref{lem_double_pat}'s proof.}
	\label{fig_lem_double_pat}
\end{figure}

Denote by $P$ the PGG defined on $G_1$ by $T$, and by $P'$ the PGG defined by $T'$ on $G'=(V',E')$ where $V'=V_1\cup V_2,\;\;E'=E\cup E_1\cup E_2$. We show that there exists an NTPNE in $P$ iff there exists one in $P'$. If there exists an NTPNE in $P$, we simply give the nodes of $G_2$ the same assignment as those of $G_1$. Since $T'$ is a double pattern of $T$, any node $v'\in V'$ must play best response, having exactly twice as many supporting neighbors than it had (or its replica had) in $P$. In addition, this assignment is clearly non-trivial as the nodes of $G_1$ have the same (non-trivial) assignment they had in $P$.

In the opposite direction, if there exists an NTPNE in $P'$, notice that for all $1\leq i\leq n$ it must be that $v_i^1$ and $v_i^2$ have identical assignments, since they both share exactly the same neighbors, and thus have identical best responses. Therefore, any node $v'\in V$ must have an even number of productive neighbors, half of which are in $V_1$ and the other half in $V_2$ (as for each productive neighbor from $V_1$ there is a respective productive neighbor from $V_2$). We then simply ignore $G_2$, and leave the assignment of $G_1$ as it is, and each node shall now have exactly half as many productive neighbors as it had in the original assignment. Since $T$ is a half pattern of $T'$, we get a PNE in $P$. Furthermore, the symmetry between matching nodes from $G_1$ and $G_2$ ensures that at least one node from $G_1$ was originally assigned 1, and so the constructed PNE in $P'$ is also non-trivial.
\end{proof}

Given Lemmas \ref{lem_half_reaches_best_shot} and \ref{lem_double_pat}, we are now able to prove Theorem \ref{thrm_all_semi_shrp}. The intuitive idea of the proof is that we halve the pattern $T_1$ (i.e. find its half-pattern) repeatedly, until eventually we reach some pattern for which we already know the equilibrium problem is hard, which, as we will see, must happen at some point. Then, by applying Lemma \ref{lem_double_pat} recursively, we have that $T_1$ is hard.

\begin{proof}(Theorem \ref{thrm_all_semi_shrp})
From Lemma \ref{lem_half_reaches_best_shot} we have that if we halve a pattern beginning with $1$ enough times, we will eventually reach the Best-Shot pattern: $T_{Best\mbox{-}Shot}[0]=1$ and $\forall k\geq 1 \;\; T_{Best\mbox{-}Shot}[k]=0$. Divide into two cases.

In the first case assume that $\forall k\in {\rm I\!N}$ it holds that $T_1[2^k]=0$. In this case, we know that no matter how many times we halve $T_1$ into patterns $T_2,T_3,...$, the value in index 1 of all these half-patterns will always be 0, i.e. $T_i[1]=0$ for all $i$. Assume that we halve $T_1$ repeatedly into patterns $T_2,T_3,...,T_m$ (where $T_{i}$ is the half pattern of $T_{i-1}$) such that $T_m$ is the first time that we reach the Best-Shot pattern. Observe $T_{m-1}$. For any \textit{even} index $k\ne 0$ it must hold that $T_{m-1}[k]=0$, otherwise $T_m$ would not be the Best-Shot pattern. Additionally, there must exist at least one \textit{odd} index $j$ s.t. $T_{m-1}[j]=1$, since $T_m$ is the \textit{first} time we reach the Best-Shot pattern. For these two reasons, we have that $T_{m-1}$ satisfies the conditions of Theorem \ref{thrm_semi_sharp_isolated_odd} and therefore NTPNE($T_{m-1}$) is NP-complete under Turing reduction. From Lemma \ref{lem_double_pat} (used inductively), we have that $\forall 1\leq i\leq m-2$ NTPNE($T_i$) is also NP-complete under Turing reduction, and specifically NTPNE($T_1$).

 In the second case, assume that there exists some $k\in {\rm I\!N}$ s.t. $T_1[2^k]=1$. In that case, after at most $k$ halvings, we reach some pattern for which the value of index 1 is 1. Assume that we halve $T_1$ repeatedly into patterns $T_2,T_3,...,T_n$ (where $T_{i}$ is the half pattern of $T_{i-1}$) such that $T_n$ is the first time that we reach a pattern for which index 1 is 1, i.e. $\forall 1\leq i \leq n-1 \;\; T_i[1]=0$ and $T_n[1]=1$. Notice that, additionally, by definition of a half-pattern for each $i$ it holds that $T_i[0]=1$ (since $T_1[0]=1$). If $T_n$ is non-monotone, then by Theorem 5 in \cite{Gilboa_Nisan_public_goods} we have that NTPNE($T_n$) is NP-complete under Turing reduction, and from Lemma \ref{lem_double_pat} (used inductively), we have that $\forall 1\leq i\leq n-1$ NTPNE($T_i$) is also NP-complete under Turing reduction, and specifically NTPNE($T_1$).
 Otherwise (i.e. $T_n$ is monotone), denote by $l$ the largest index s.t. $T_n[l]=1$, and observe $T_{n-1}$. By definition of double-patterns, we have that:
 \begin{align*}
 \forall j\in {\rm I\!N} \;\; T_{n-1}[2j]=
    \begin{cases}
        1 & \text{if $j\leq l$}\\
        0 & \text{otherwise}
    \end{cases}
\end{align*}
 
 i.e. the value in the even indices up to $2l$ is 1, and afterwards 0. Since $T_n$ is defined to be the first halving of $T_1$ s.t. its value in index 1 is 1, we have that $T_{n-1}[1]=0$. However, since the definition of a double-pattern does not restrict its values in odd indices, there might be odd indices (strictly larger than 1) for which the value of $T_{n-1}$ is 1. Divide into 3 sub-cases:

 \textit{Sub-case 1:} If there exists some $z\leq l$ s.t. $T_{n-1}[2z+1]=1$, then by Lemma \ref{lem_alternating_with_odd}, we have that NTPNE($T_{n-1}$) is NP-complete under Turing reduction.
 
 \textit{Sub-case 2:} Otherwise, if there exists some $z>l$ s.t. $T_{n-1}[2z+1]=1$, then observe the pattern $T_{n-1}'$, which we define as the pattern shifted left by $2l$ from $T_{n-1}$ i.e.:
 \[\forall j\in {\rm I\!N} \;\; T_{n-1}'[j]=T_{n-1}[j+2l]\]
 Notice that this pattern satisfies the conditions of Theorem \ref{thrm_semi_sharp_isolated_odd}, and therefore NTPNE($T_{n-1}'$) is NP-complete under Turing reduction. Then, by applying Theorem 7 from \cite{Gilboa_Nisan_public_goods} $l$ times, we have that NTPNE($T_{n-1}$) is also NP-complete under Turing reduction.
 
 \textit{Sub-case 3:} Otherwise (i.e. there is no odd index whatsoever in which the value of $T_{n-1}$ is 1), then by Corollary \ref{cor_alternating} we have that NTPNE($T_{n-1}$) is NP-complete under Turing reduction.
 
 And so, in either case we have that NTPNE($T_{n-1}$) is NP-complete under Turing reduction, and therefore from Lemma \ref{lem_double_pat} (used inductively), we have that $\forall 1\leq i\leq n-1$ NTPNE($T_i$) is also NP-complete under Turing reduction, and specifically NTPNE($T_1$).
 \end{proof}

\section{Hardness of All Spiked Patterns}
\label{sec_all_spiked}
There are several finite, spiked patterns that we have not yet proved hardness for, and we now have enough tools to close the remaining gaps. We remind the reader that spiked patterns are patterns that begin with 1,0,1. The following theorem formalizes the result of this section, and completes the characterization of all finite patterns.

\begin{theorem}
\label{thrm_all_spiked}
Let $T$ be a finite, spiked BRP. Then NTPNE($T$) is NP-complete under Turing reduction.
\end{theorem}

The intuitive idea of the proof is as follows. If the pattern simply alternates between 1 and 0 a finite amount of times (and at least twice, since the pattern is spiked), followed by infinite 0's, i.e. the pattern is of the form
\[T=[\underbrace{1,0,1,0,1,0,...,1,0}_{finite\;number\;of\;'1,0'},0,0,0,...]\]
then the problem\footnote{In fact, Corollary \ref{cor_alternating} gives a more general result, but we currently only need the special case where the pattern ends with infinite 0's.} is already shown to be hard by Corollary \ref{cor_alternating}. Otherwise, we wish to look at the first "disturbance" where this pattern stops alternating from 1 to 0 regularly. Either the first "disturbance" is a 1 at an odd index, i.e. the pattern is of the form
\[
T=[\underbrace{1,0,1,0,1,0,...,1,0}_{finite\;number\;of\;'1,0'},1,\textcolor{red}{1},?,?,...]
\]
or the first "disturbance" is a 0 at an even index, i.e. the pattern is of the form
\[
T=[\underbrace{1,0,1,0,1,0,...,1,0}_{finite\;number\;of\;'1,0'},\textcolor{red}{0},?,?,...,1,?,?,...]
\]
(in the latter option, after the first "disturbance" there must be some other index with value 1, since otherwise the pattern fits the form of Corollary \ref{cor_alternating}). The first option was solved in Lemma \ref{lem_alternating_with_odd}, and the second option can be solved using our previous results, as we shall now formalize in the proof. 

\begin{proof} (Theorem \ref{thrm_all_spiked})
If $T$ satisfies the conditions of Corollary \ref{cor_alternating} or Lemma \ref{lem_alternating_with_odd} then NTPNE($T$) is NP-complete under Turing reduction according to them. Otherwise, let $k$ be the smallest integer such that $T[2k]=0$. Denote by $T'$ the pattern which is shifted left by $2k-2$ from $T$, i.e.:
\[\forall j\geq 0\;\; T'[j]=T[j+2k-2]\]
Notice that from definition of $k$ (being the first even index such that $T[2k]=0$) we have that for all $j<k$ it holds that $T[2j]=1$. Moreover, since $T$ does not satisfy the conditions of Lemma \ref{lem_alternating_with_odd} it must hold for all $j\leq k$ that $T[2j-1]=0$, i.e. the value of $T$ in the odd indices until $2k$ is 0 (since otherwise $T$ would start with a finite number of 1,0, followed by two consecutive 1's, and would satisfy the conditions of Lemma \ref{lem_alternating_with_odd}). Thus, we have that

\begin{equation}
\label{eqn:01}
\forall j<2k \;\; T[j]=
    \begin{cases}
        1 \text{ if $j$ is even}\\
        0 \text{ if $j$ is odd}
    \end{cases}
\end{equation}

In particular, we have that $T[2k-2]=1,\; T[2k-1]=0$, which implies that $T'[0]=1,\;T'[1]=0$; as $T[2k]=0$ we have that $T'[2]=0$, and thus we conclude that $T'$ is semi-sharp. In addition, since $T$ does not satisfy the conditions of Corollary \ref{cor_alternating}, there must be some other index $x>2k$ such that $T[x]=1$, and therefore we have that $T'$ is non-monotone. Therefore, by Theorems \ref{thrm_semi_sharp_isolated_odd} and \ref{thrm_all_semi_shrp}, we have that NTPNE($T'$) is NP-complete under Turing reduction. We now wish to use this in order to prove that NTPNE($T$) is also hard.

From Equation \ref{eqn:01}, we can apply Theorem 7 of \cite{Gilboa_Nisan_public_goods} $(k-1)$ times (we remind the reader that this theorem states that prefixing a hard pattern with $1,0$ maintains its hardness), and we have that NTPNE($T$) is NP-complete under Turing reduction.
\end{proof}

\section*{Acknowledgements}
I would like to thank Noam Nisan for many useful conversations, and for suggesting the Copy Gadget in the proof of Theorem \ref{thrm_0_or_2}.
I would like to thank Roy Gilboa for many useful conversations, and for adjusting the Copy Gadget in the proof of Theorem \ref{thrm_0_or_2}.
I would like to thank Noam Nisan for communicating to me the alternative solution to the monotone case (see footnote \ref{Sigal_Oren_footnote}), which was suggested by Sigal Oren.
I would like to thank the anonymous ICALP reviewers for their helpful feedback.

\bibliographystyle{splncs04}
\bibliography{public_goods.bib}

\end{document}